%% file: AcceleratedMeasuredArticle.tex
\RequirePackage{etoolbox}
\newtoggle{article}

 \toggletrue{article} 

\newcommand{\inArticle}[1]{\iftoggle{article}{#1}{}} 
\newcommand{\inThesis}[1]{\iftoggle{article}{}{#1}}  

\newcommand{\sectionInShort}[1]{
	\inThesis{\subsection{#1}}
	\inArticle{\section{#1}}
}

\newcommand{\subSectionInShort}[1]{
	\inThesis{\subsubsection{#1}}
	\inArticle{\paragraph{#1}}
}

\documentclass[11pt]{paper}
\usepackage{geometry}                
\usepackage{graphicx}
\usepackage{amssymb}
\usepackage{epstopdf}
\usepackage{amsthm}
\usepackage{amssymb} 
\usepackage[cmex10]{amsmath}

\usepackage{accents}

\newcommand{\ubar}[1]{\underaccent{\bar}{#1}}
 
\newtheorem{theorem}{Theorem}[section]
\newtheorem{lemma}[theorem]{Lemma}

\newtheorem{corollary}[theorem]{Corollary}

\newcommand{\argmax}{\operatornamewithlimits{argmax}}

\renewcommand{\subsection}{\section}

\usepackage[ruled, vlined]{algorithm2e}

\title{Fast Non-Monotone Submodular Maximisation Subject to a Matroid Constraint}
\author{Pau Segui-Gasco
\thanks{p.seguigasco@cranfield.ac.uk}
\thanks{Centre for Autonomous and Cyber-Physical Systems, Institute of Aerospace Sciences, School of Aerospace Transport and Manufacturing, Cranfield University, Bedfordshire, UK.} 
\and Hyo-Sang Shin \footnotemark[2]
}

\begin{document}
\maketitle

\abstract{In this work we present the first practical $\left(\frac{1}{e}-\epsilon\right)$-approximation algorithm to maximise a general non-negative submodular function subject to a matroid constraint. Our algorithm is based on combining the decreasing-threshold procedure of Badanidiyuru and Vondrak (SODA 2014) with a smoother version of the measured continuous greedy algorithm of Feldman et al. (FOCS 2011). This enables us to obtain an algorithm that requires $O(\frac{nr^2}{\epsilon^4} \big(\frac{\bar{d}+\ubar{d}}{\bar{d}}\big)^2 \log^2({\frac{n}{\epsilon}}))$ value oracle calls, where $n$ is the cardinality of the ground set, $r$ is the matroid rank, and $\ubar{d}, \bar{d} \in \mathbb{R}^+$ are the absolute values of the minimum and maximum marginal values that the function $f$ can take i.e.: $ -\ubar{d} \leq f_S(i) \leq \bar{d}$, for all $i\in E$ and $S\subseteq E$, where $E$ is the ground set. The additional value oracle calls with respect to the work of Badanidiyuru and Vondrak come from the greater spread in the sampling of the multilinear extension that the possibility of negative marginal values introduce.}

\

\input{acceleratedMeasuredGreedy.tex}
\input{AcceleratedMeasuredArticle.bbl}


\end{document}

%% file: acceleratedMeasuredGreedy.tex


\inArticle{

\sectionInShort{Introduction}


Submodular maximisation problems have drawn a lot of attention recently \cite{krause2012submax}. This interest is due to a good confluence of theoretical results and practical applicability. Intuitively, a submodular set function is said to be so because it exhibits diminishing marginal returns, i.e.: the marginal value that an element adds to a set decreases as the size of the set increases. This simple property arises naturally in many applications and is what enables good theoretical tractability.  It has been used in a variety of application domains, to name but a few: markets \cite{lehmann2001comb, dughmi2009revenue}, influence in networks \cite{kempe2003social}, document summarisation \cite{lin2010multi}, and sensor placement \cite{krause2008efficient, leskovec2007cost}. In particular, this paper focuses on matroid-constrained submodular maximization of a general non-negative submodular function.

Matroids are an incredibly powerful abstraction of independence. They capture seemingly disconnected notions such as linear independence, forests in graphs, traversals, among many others. Interestingly, with linear sum objectives (modular function), they are inextricably linked with the greedy algorithm. If the greedy algorithm is optimal, then there is an implicit matroid; if there is a matroid then the greedy algorithm is optimal \cite{Schrijver2003}. They provide a flexible framework to characterise a variety of relevant constraints such as: partitions, schedules, cardinality, or even rigidity.

Importantly, the combination of a submodular function and a matroid constraint plays a unifying role for many well-known combinatorial optimisation problems such as: Submodular Welfare (also known as Submodular Task Allocation), Max $k$-Cover, Max Generalised Assignment, Max Facility Location, and Constrained Max Cut (e.g. Max Bisection) among others.


However, maximising a submodular function subject to a matroid constraint is NP-Hard. Hence much of the research effort has focused on obtaining good approximation algorithms. A classic result by Nemhauser, Wolsey, and Fisher \cite{Nemhauser1978} shows that the Greedy Algorithm is a $\frac{1}{2}$-approximation for non-negative monotone submodular functions. More recently, a fruitful avenue of research was spurred by Vondrak \cite{Vondrak2008, Calinescu2011} by showing that solving a relaxation of the problem based on the multilinear extension using the Continuous Greedy Algorithm and then rounding the result yields good approximation algorithms. They present their result in the context of non-negative monotone submodular functions to yield a $ (1-\frac{1}{e})$-approximation. Shortly after, Feldman et al. \cite{Feldman2011} modified Vondrak's algorithm to develop the Measured Continuous Greedy Algorithm which supported both non-negative non-monotone and monotone submodular functions. Their algorithm was the first to achieve a $\frac{1}{e}$-approximation for general non-negative submodular functions subject to a matroid constraint. Both continuous greedy algorithms can find the aforementioned approximation bounds for solving relaxation problems subject to the more general constraint class of down-closed solvable polytopes. An important breakthrough in this field are Contention Resolution Schemes, a rounding framework proposed in \cite{vondrak2011submodular}, because they provide a paradigm for developing rounding schemes for a combination of useful constraints including matroids and knapsacks. Thus the combination of continuous greedy relaxations and Contention Resolution schemes enabled approximation algorithms for many important submodular maximization problems.

However, continuous greedy algorithms are impractically slow. Vondrak's algorithm required  $\Theta(n^8)$ value oracle calls, while Feldman's Measured Continuous Greedy requires $O(n^6)$ assuming oracle access to the multilinear extension, which needs to be sampled in general, creating additional overhead, possibly around $O(n^3)$ or $O(n^2)$. 

To remedy this, Badanidiyuru and Vondrak \cite{FastVondrak} proposed an efficient $ \left(1-\frac{1}{e}- \epsilon \right)$-approximation algorithm that uses $O\left(\frac{nr}{\epsilon^4}\log^2{\frac{n}{\epsilon}}\right)$ value oracle calls, for non-negative monotone submodular function maximisation subject to a matroid constraint. Where $n$ is the cardinality of the ground set, and $r$ is the rank of the matroid. They achieve this by using a Decreasing-Threshold procedure that enables a reduction of both the number of steps and the number of samples needed at each step of the continuous greedy process.

In the inapproximability front, for matroid constrained non-negative monotone submodular maximisation Feige \cite{feige1998threshold} showed that improving over the $1-\frac{1}{e}$ threshold is NP-Hard, and so the continuous greedy guarantees are tight. In the non-monotone case, \cite{gharan2011submodular} showed that no polynomial time algorithm can achieve an approximation better than 0.478. Closing the gap between $\frac{1}{e}$ and 0.478 remains an important open problem where recent advances have been made: Ene and Nguyen in \cite{ene2016constrained} give a 0.372-approximation, while  Feldman et al in \cite{buchbinder2016constrained} improve it to a $0.385$-approximation. This improvements are relatively small, recall that $\frac{1}{e} \approx 0.368$, but show that there is room for future improvement. Both algorithms are based on the measured continuous greedy, and would therefore benefit from the techniques presented here.

Our contribution in this work is the first $\left(\frac{1}{e}-\epsilon \right)$-approximation algorithm to maximise a general non-negative submodular function subject to a matroid constraint with a practical time complexity. Our algorithm is based on combining the Decreasing-Threshold procedure of \cite{FastVondrak} with a smoother version of the measured continuous greedy algorithm of \cite{Feldman2011}. This allows us to obtain an algorithm that requires $O(\frac{nr^2}{\epsilon^4} \big(\frac{\bar{d}+\ubar{d}}{\bar{d}}\big)^2 \log^2({\frac{n}{\epsilon}}))$ value oracle calls, where $\ubar{d}, \bar{d} \in \mathbb{R}^+$ are the absolute values of the minimum and maximum marginal values that the function $f$ can take, i.e. $ -\ubar{d} \leq f_S(i) \leq \bar{d}$, for all $i\in E$ and $S\subseteq E$. The additional oracle calls with respect to \cite{FastVondrak} come from the additional spread in the sampling of the multilinear extension that the possibility of negative marginal values introduces.

\subSectionInShort{Definitions}
Let us now introduce the definitions of a matroid and a submodular function.
A function $f:2^E \rightarrow \mathbb{R^+}$ on a set $E$ is said to be submodular if:
\begin{equation}
f(A)+f(B)\geq f(A\cup B) + f(A\cap B), \textit{ for all } A,B \subseteq E. 
\end{equation}
A more intuitive, but equivalent, definition can be formulated in terms of the marginal value added by an element: given $Y,X\subseteq E$ satisfying $X\subseteq Y$ and $x\in E \setminus Y$, then $f(X+x)-f(X)\geq f(Y + x)-f(Y)$. Herein we overload the symbol $+$ ($-$) to use it as shorthand notation for the addition (subtraction) of an element to a set, i.e. $S + i = S \cup \{i\}$ ($S - i = S \setminus \{i\}$). A matroid can be defined as follows \cite{Schrijver2003}: a pair $\mathcal{M} = (E,\mathcal{I})$ is called a matroid if $E$ is a finite set and $\mathcal{I}$ is a nonempty collection of subsets of E satisfying:  $\emptyset \in\mathcal{I}$; if $A\in\mathcal{I}$ and $B\subseteq A$, then $B\in\mathcal{I}$; and if $A,B\in\mathcal{I}$ and $|A|<|B|$, then $A + z \in \mathcal{I}$ for some $z \in B \setminus A$.

A matroid base $B \subseteq E$ is a maximally independent set $B \in \mathcal{I}$, that becomes dependent by adding an additional element $e\in E\setminus B$, i.e. $B+e \notin \mathcal{I}$. A key property of matroids is that we can exchange elements between bases, let $\mathcal{B}$ denote the set of all bases of a matroid $\mathcal{M}$, then the exchange property can be is formalised as follows:
\begin{lemma} (Corollary 39.12A in \cite{Schrijver2003}) 
Let $\mathcal{M} = (N,\mathcal{I}) $ be a matroid, and $B_1,B_2\in \mathcal{B}$ be two bases. Then there is a bijection $\phi : B_1 \rightarrow B_2$ such that for every $b\in B_1$ we have $B_1-b + \phi(b) \in \mathcal{B}$.
\label{basebijection}
\end{lemma}

Now we can formally define our target problem:\\
\textit{Given a ground set $E$, a matroid defined upon it $\mathcal{M} = ($E$,\mathcal{I})$, and a general non-negative submodular function $f:2^{E} \rightarrow \mathbb{R}^{+}$, find:}
\begin{equation}
\max_{S\in\mathcal{I}} f(S).
\end{equation}

A common approach to solve combinatorial optimisation problems is by using relaxation and rounding. This involves solving a continuous problem, in which we allow fractional solutions, and a rounding procedure which transforms the fractional output back to a discrete solution. Now let us explain the relaxation that is used in continuous greedy algorithms. There are two key elements needed to define the relaxation: the domain that contains fractional solutions, and a function to evaluate fractional solutions. In our problem, the most natural representation for the domain is the matroid polytope $P(\mathcal{M})$, which is the convex hull of the incidence vectors of the independent sets in the matroid, i.e.: $P(\mathcal{M})= \text{Conv}(\{\mathbf{1}_S: \forall S \in \mathcal{I}\})$ \cite{Schrijver2003}. Here $\mathbf{1}_S$ denotes the incidence vector of a set $S\subseteq E$, which contains a 1 for each element in $S$ and 0 elsewhere, thus $\mathbf{1}_S\in[0,1]^{E}$. The function for evaluation of fractional solutions is the multilinear extension: given a set $E$ and a submodular set function $f: 2^E \rightarrow \mathbb{R}^{+}$, its multilinear extension $F: [0,1]^E \rightarrow \mathbb{R}^{+}$ is defined as:

\begin{equation}
F(\mathbf{y})\triangleq\mathbb{E}[f(R(\mathbf{y}))]=\sum_{\mathcal{S}\subseteq E}{f(\mathcal{S})\prod_{i\in \mathcal{S}}y_i}\prod_{j\notin \mathcal{S}}(1-y_j),
\end{equation}
where $R(\mathbf{y})$ is a random set containing each element $i\in E$ with probability $y_i$. Additionally, $F$ can be lower bounded in terms of $f$ :
\begin{lemma} (from \cite{Feldman2011})
Let $f:2^E\rightarrow \mathbb{R}^{+}$ be a submodular function; let $F:[0,1]^E\rightarrow \mathbb{R}^{+}$ be the multilinear extension of $f$; and let $\mathbf{y}\in[0,1]^E$ be a vector whose components are bounded by $a$, i.e. $y_i \leq a,  \forall i \in E$. Then, for every $S \subseteq{E}$, we have that 
\begin{equation}
F(\mathbf{y} \vee \mathbf{1}_S) \geq (1-a)f(S).
\end{equation}
\label{boundF}
\end{lemma}
A key magnitude used in the Measured Continuous Greedy Algorithm is the marginal value of an element $e\in E$ given a fractional solution $\mathbf{y} \in P(\mathcal{M})$, we refer to it by $\Delta F_{e}(\mathbf{y})$, defined as follows:
\begin{equation}
\Delta F_e(\mathbf{y}) \triangleq \mathbb{E}[f_{R({\mathbf{y}})}(e)]= F(\mathbf{y} \vee \mathbf{1}_e) - F(\mathbf{y}).
\end{equation}

In other words, it is the value that would be created by adding the full element $i$ to the fractional solution $\mathbf{y}$. Here we have used the common shorthand notation for the marginal value of an element $f_S(e) = f(S+e)-f(S)$.

In general, given an arbitrary non-negative submodular function there is no closed form of multilinear extension that enables us to evaluate it efficiently. The usual way to deal with this is to sample it, and so we need the following concentration inequality that enables us to bound the sampling error:

\begin{lemma} (Hoeffding Bound, Theorem 2 in \cite{Hoeffding}) 
Let $X_1, ..., X_m$ be independent random variables such that for each $i$, $a \leq X_i \leq b $, with $a,b \in \mathbb{R}$. Let $\tilde{X}=\frac{1}{m}\sum_{i=1}^{m} X_i$. Then
\begin{equation*}
\Pr[|\tilde{X} - \mathbb{E}(X)| > t ]\leq 2e^{-\frac{ 2  t^2}{(b-a)^2}m} .
\end{equation*}
\label{hoeffding}
\end{lemma}

In this work we shall not dwell on the rounding step. Suffice it to say that there exist algorithms, such as Contention Resolution Schemes \cite{Chekuri2014}, Pipage-Rounding \cite{Ageev2004, Calinescu2011, vondrak2013symmetry}, or Swap-Rounding \cite{vondrak2011submodular}, that given a general non-negative submodular function $f$, its multilinear extension $F$, and a point $\mathbf{y}\in P(\mathcal{M})$, find a set $S\in\mathcal{I}$, such that $f(S)\geq F(\mathbf{y})$. The most efficient technique is Swap-Rounding, and its results are used in \cite{FastVondrak}. However, the swap-rounding results in \cite{vondrak2011submodular} hinge around Chernoff-like concentration bounds for monotone submodular functions, but in light of the concentration bounds for non-monotone submodular functions in \cite{Vondrak2010}, we believe that this technique can be adapted to work with non-monotone submodular functions.

}

\sectionInShort{Algorithm}


\begin{algorithm}
	\SetKwInOut{Input}{Input}
	\SetKwInOut{Output}{Output}
	\SetKwInOut{Return}{Return}
	\SetKwFunction{DecreasingThreshold}{Decreasing-Threshold}
	\DontPrintSemicolon
	
	\Input{  $f:2^{E} \rightarrow \mathbb{R}^{+}$,  $\epsilon \in [0,1]$, $\mathcal{I} \subseteq 2^E$.} 
	\Output{ A point $\mathbf{y}\in P(\mathcal{M})$, such that $F(\mathbf{y})\geq (\frac{1}{e}-2\epsilon)f(OPT)$.}
	\BlankLine
	\tcp{Initialisation}
    $\mathbf{y}(0)\gets \mathbf{0}$\;
    $\delta \gets \epsilon$\;
    \BlankLine
    \tcp{Main Loop}
   \For{$t = \{0,\delta,2\delta,3\delta,\dots, 1-\delta\}$}{
		$B(t) \gets \DecreasingThreshold{$f,\mathbf{y}(t),\epsilon, \delta, \mathcal{I}$}$\;
		\For{$i\in E$}{
			\eIf{$i\in B(t)$}{
				 $y_i(t+\delta) \gets  1 + e^{- \delta}(y_i(t)-1)$\;
			}{ 
				 $y_i(t+\delta) \gets  y_i(t)$\;
			}
		}
   }
   \Return{ $\mathbf{y}(1)$ }
\caption{Accelerated Measured Continuous Greedy}
\label{algoThresholdMeasured}
\end{algorithm}

\begin{algorithm}
\caption{Decreasing-Threshold}
	\SetKwInOut{Input}{Input}
	\SetKwInOut{Output}{Output}
	\SetKwInOut{Return}{Return}
	
	\Input{ $f:2^{E} \rightarrow \mathbb{R}^{+}$,  $\mathbf{y} \in [0,1]^E$, $\epsilon \in [0,1]$, $\delta \in [0,1]$, $\mathcal{I} 	\subseteq 2^E$.}
	\Output{A set $B\subseteq E$, such that $B\in\mathcal{I}$.}
	\BlankLine
	\tcp{Initialisation}
   $B \gets \emptyset$\;
   $\bar{d} \gets \max_{i\in E} f(i)$\;
   $\bar{y}' \gets \max_{i\in E}(1+e^{-\delta}(y_i-1))$\;
   \BlankLine
   \tcp{Main Loop}
   \For{($w = d$; $w\geq \epsilon \frac{\bar{d}}{r}(1-\bar{y}')$; $w \gets w(1-\epsilon)$)}{
   		\For{$e\in E$}{
      		$ w_{e}(B,\mathbf{y}) \gets \Delta F_e(\mathbf{y}(B,\delta))$, \\
      		\tcp{\small averaging  $O\bigg(\frac{r^2}{\epsilon^2}\big(\frac{\bar{d}+\ubar{d}}{\bar{d}}\big)^2\log(|E|)\bigg)$  iid random samples.}
      \If{$w_{e}(B,\mathbf{y})\geq w$ and $B+e \in \mathcal{I}$}{
      		$B \gets B + e$
      	}
		}

   }
   \Return{ $B$} 
   {\small*Note that the notation $\mathbf{y}(B,\delta)$ means $y_i(B,\delta)=y_i(t)$ for $i\notin B$, and $y_i(B,\delta)=1+e^{-\delta}(y_i-1)$ for $i\in B$.}
\label{algoThreshold}
\end{algorithm}


\inArticle{
The algorithm we present here is based on the Decreasing-Threshold procedure that enabled Badanidiyuru and Vondrak \cite{FastVondrak} to achieve a very efficient algorithm, $O(\frac{nr}{\epsilon^4}\log^2{\frac{n}{\epsilon}})$, for maximising a non-negative monotone submodular function subject to a matroid constraint. The continuous greedy algorithm finds an approximate solution to the relaxation, $\mathbf{y} \in P(\mathcal{M})$, by integrating between $t=0$ to $t=1$ the differential equation $\frac{d\mathbf{y}}{dt} = \argmax_{v \in P(\mathcal{M})} \sum_{i\in E}\Delta F_i(\mathbf{y})v_i$. The algorithm in \cite{FastVondrak} is basically the standard continuous greedy algorithm, but uses a much more efficient Decreasing-Threshold procedure to find the maximum marginal improvement direction, i.e. $v_i \, \forall i \in E$ in the equation above, in each iteration. To enable an approximation algorithm for the non-monotone case Feldman \cite{Feldman2011} proposed the measured continuous greedy algorithm, which similarly integrates $\frac{d\mathbf{y}(t)}{dt} = (1-\mathbf{y}(t))\odot\argmax_{v \in P(\mathcal{M})} \sum_{i\in E}\Delta F_i(\mathbf{y})v_i$ ($\odot$ represents element by element multiplication here).

The advantage of the Decreasing-Threshold procedure is that it requires a smaller number of integration steps and a smaller number of samples per step. Here we apply its key idea to the measured continuous greedy. However, we need to perform two crucial modifications. The first modification is a smoother version of the integration rule in the measured continuous greedy algorithm of Feldman et al. \cite{Feldman2011}. We use an update step that makes the algorithm more continuous-like. Instead of using the integration step $y(t+\delta)= y(t) + \delta(1-y(t))$ as proposed by \cite{Feldman2011}, we increment the solution at a rate of $\frac{dy(t)}{dt} = 1-y(t)$ by using the integration rule: $y(t+\delta) =  1 + e^{- \delta}(y(t)-1)$, see Algorithm \ref{algoThresholdMeasured}. This allows us to smooth the tradeoff between error and running time. The second modification is in the sampling error up to which we estimate the marginal values of the multilinear extension in the Decreasing-Threshold procedure. Badanidiyuru and Vondrak \cite{FastVondrak} used an additive and multiplicative error bound which enabled them to use only $O(\frac{1}{\epsilon^2}r\log(n))$ samples. However, since our algorithm is meant to work with non-monotone functions, we may have negative marginal values, this increases the number of samples required for two reasons: first, it prevents us from using a multiplicative-additive error bound; and second the marginal values, when sampled, have inherently more spread. To quantify our error we use Hoeffding's concentration inequality, which forces us to use $O\big(\frac{r^2}{\epsilon^2}\big(\frac{\bar{d}+\ubar{d}}{\bar{d}}\big)^2\log(n)\big)$, where $\ubar{d}, \bar{d} \in \mathbb{R}^+$ are the absolute values of the minimum and maximum marginal values that the function $f$ can take, i.e. $ -\ubar{d} \leq f_S(i) \leq \bar{d}$, for all $i\in E$ and $S\subseteq E$. (Recall that if $f$ were to be monotone we would have $\ubar{d}=0$.) The resulting Decreasing-Threshold procedure is presented in Algorithm \ref{algoThreshold}, and it takes an additional $O(r \big(\frac{\bar{d}+\ubar{d}}{\bar{d}}\big)^2)$ value oracle calls. Using only an additive bound instead of an additive and multiplicative introduces an extra $O(r)$ factor in the number of samples required per evaluation of the multilinear extension. While we also incur the additional $O\big(\big(\frac{\bar{d}+\ubar{d}}{\bar{d}}\big)^2\big)$ term to cope with the larger spread that is introduced by the possibility of negative marginal values. This results in an algorithm that finds an $\frac{1}{e} - \epsilon$ approximation with a number of value oracle calls of $O(\frac{nr^2}{\epsilon^4} \big(\frac{\bar{d}+\ubar{d}}{\bar{d}}\big)^2 \log^2({\frac{n}{\epsilon}}))$. 
It is not as efficient as the algorithm for the monotone case, but to the best of our knowledge, this constitutes the first practical $\frac{1}{e}-\epsilon$ approximation algorithm for general non-negative submodular function subject to a matroid constraint. }

\sectionInShort{Algorithm Analysis}

We split the analysis of the accelerated measured continuous greedy algorithm in three parts: first we show that the solution produced is feasible, i.e. $\mathbf{y}(1) \in P(\mathcal{M})$; then we prove the $\frac{1}{e}-\epsilon$ approximation; and finally we study its running time.

\subSectionInShort{Feasibility}

\begin{theorem}
The accelerated measured continuous greedy algorithm produces a feasible fractional solution, i.e., $\mathbf{y}(1) \in P(\mathcal{M})$.

\label{yinkDT}
\end{theorem}

\begin{proof}
We follow the approach used by \cite{Feldman2011}. We first define a vector $\mathbf{x}$ that coordinate wise upper-bounds $\mathbf{y}(1)$. Then, given that $P(\mathcal{M})$ is down-monotone, we only need to show that $\mathbf{x}$ is in $P(\mathcal{M})$ to show that $\mathbf{y}(1) \in P(\mathcal{M})$. Consider the vector $\mathbf{x}=\delta\sum_{l=0}^{\frac{1}{\delta}-1}  \mathbf{1}_{B(\mathbf{y}(l\delta))}$. This is a coordinate-wise upper bound of $\mathbf{y}(1)$ because when $i\in B$, we have that $y_i(t+\delta)-y_i(t) = 1 + e^{- \delta}(y_i(t)-1) -y_i(t) = (1-e^{-\delta})(1-y_i(t))\leq 1-e^{-\delta}\leq \delta$, for all $ \delta \in [0,1]$; and when $i\notin B$ $y_i(t+\delta)-y_i(t) = 0$. We now show that $\mathbf{x}$ is in $P(\mathcal{M})$. First, note that by definition $\mathbf{1}_B \in P(\mathcal{M})$. Then, observe that $\mathbf{x}/\delta$ is the sum of $\frac{1}{\delta}$ points in $P(\mathcal{M})$, thus $(\mathbf{x}/\delta)/(1/\delta)=\mathbf{x}$ is a convex combination of points in $P(\mathcal{M})$, hence $\mathbf{x}\in P(\mathcal{M})$, and consequently $\mathbf{y}(1) \in P(\mathcal{M})$.
 \end{proof}

\inArticle{
In fact, it is possible to find a point that still lies in $P(\mathcal{M})$ even with a $t>1$, specifically stopping at a value depending on a magnitude called the density of the matroid. This yields tighter approximation bounds that match those found for particular matroids, such as partition matroids. Asymptotically, however, these bounds are the same as those presented here. Therefore, in the aim of simplicity, the analysis is carried out with a stopping time of 1. The interested reader is referred to \cite{Feldman2011}.}

Now let us establish a bound to the coordinates of $\mathbf{y}(t)$ that will become useful later in the analysis of the approximation ratio.

\begin{lemma}
At time $0\leq t \leq1$, we have that:
\begin{equation}
y_{i}(t)\leq 1-e^{-t}, \quad \forall  i \in E.
\end{equation}
\label{boundY}
\end{lemma}
\begin{proof}
Consider the recurrence $g(n+1)=1+e^{-\delta}(1-g(n))$, with $g(0)=0$. This recurrence has the following solution:  $g(n) = 1-e^{-n\delta}$. Now, in our algorithm $t$ is incremented linearly, so the number of iteration, $n$, and $t$ are related by $t=n\delta$. At  $t$, all the coordinates of $\mathbf{y}$ either, they stay constant, i.e. $y_i(t+\delta)=y_i(t)$, when $i\notin B$, or increase, i.e. $y_i(t+\delta) = 1 + e^{-\delta}(y_i(t) - 1)$, when $i\in B$. Hence, given that $g(n)$ is non-decreasing, the recurrence $g$ is an upper bound because it corresponds to incrementing in each and every iteration. Consequently, at $t$, $y(t) \leq g(\frac{t}{\delta})= 1-e^{-t}$.
\end{proof}

\subSectionInShort{Approximation Ratio}

First we show the gain that the algorithm makes in a single step, and then use this to build a recurrence relation that yields the approximation ratio. But before we do that, we bound the error introduced by sampling:

\begin{corollary} 
Given a non-negative submodular function $f: 2^E\rightarrow \mathbb{R}^+$, and a point $\mathbf{y}\in[0,1]^E$;
let $\ubar{d}, \bar{d} \in \mathbb{R}^+ $ be the minimum and maximum marginal values of $f$, such that $-\ubar{d} \leq f_S(j) \leq \bar{d} $ for all $S\subseteq E$ and $j\in E$; let $R_1, R_2, ..., R_m$ be iid samples drawn from $R(\mathbf{y})$, let $w_j(\mathbf{y}) = \frac{1}{m} \sum_{i=1}^{m} f_{R_i}(j)$; and let $f(OPT)=\max_{S\in \mathcal{I}}f(S)$. Then,
\begin{equation*}
\begin{aligned}
\Pr(|w_j(\mathbf{y})- \Delta F_j(\mathbf{y})| \geq \beta f(OPT) ) \leq 2e^{-2m\beta^2\big(\frac{\bar{d}}{\bar{d }+ \ubar{d}}\big)^2}.
\end{aligned}
\end{equation*}
\label{hoeffdingNonMon}
\end{corollary}

\begin{proof}
Immediate application from the Hoeffding bound in Lemma \ref{hoeffding}, noting that $f(OPT) = \max_{S\in\mathcal{I}}f(S) \geq \max_{e\in E} f(\{e\}) =\bar{d}$.

\end{proof}

Now we can present the improvement made by the algorithm in a single step:

\begin{lemma}
Let $OPT$ be an optimal solution. Given a fractional solution $\mathbf{y}$, the Decreasing-Threshold produces a set $B$ such that, with $\mathbf{y}' = \mathbf{1} + e^{- \mathbf{1}_B\delta}\odot(\mathbf{y}-\mathbf{1}) $, we have:
\begin{equation}
F(\mathbf{y}') - F(\mathbf{y})  \geq (1-e^{-\delta}) \bigg( (1-4\epsilon)  (1-\bar{y}')f(OPT) - F(\mathbf{y}') \bigg)\\
\end{equation}
\label{deltaValDT}
\end{lemma}

\begin{proof}
This proof follows closely the proof of Claim 4.1 in \cite{FastVondrak} but with several modifications to avoid assuming monotonicity, namely: the stopping threshold, the sampling error, and the increment bound.
Assume that the Decreasing-Threshold procedure returns a sequence of $r$ elements $B = \{b_1,b_2, \dots, b_r \}$, indexed in the order in which they were chosen. Let $O = \{o_1, o_2, \dots, o_r \}$ be an optimal solution indexed as per the exchange property of the matroids in Lemma \ref{basebijection} such that $\phi(b_i)= o_i$. Additionally, let $B_i$ and $O_i$ denote the first $i$ elements of $B$ and $O$ respectively, i.e. $B_i$ is the sequence in which the elements have been added to $B$ in algorithm \ref{algoThreshold} up until the $i$th element was added. If the procedure returns fewer than $r$ elements or the optimal solution contained fewer than $r$ elements, formally we just add dummy elements with value 0, so that $|B|=r$ and $|O| = r$.

Now let us bound the marginal values of the elements selected by the Decreasing-Threshold procedure with respect to those in the optimal solution. Recall that $\mathbf{y}(S,\delta)$ is the notation that we use to refer to the point such that $y_k(S,\delta)=y_k$ if $k\notin S$ and $y_k(S,\delta)=1 + e^{- \delta}(y_k-1)$ if $k\in S$. When $b_i$ is selected, let $w$ be the current threshold, hence $w_{b_i}(\mathbf{y}(B_{i-1},\delta))\geq w$. At this point $o_i$ is a candidate element, thus we have one of two situations depending on whether the procedure has finished: if the threshold has not dropped below $\frac{\epsilon}{r} d 
(1-\bar{y}')$, the value of $w_{o_i}(\mathbf{y}(B_{i-1},\delta))$ must be below the threshold in the previous iteration, i.e. $w_{o_i}(\mathbf{y}(B_{i-1},\delta)) \leq \frac{w}{(1-\epsilon)}$ (otherwise it would have been chosen already); conversely, if the procedure has terminated, $b_i$ is a dummy element with value 0, and the value of $w_{o_i}(\mathbf{y}(B_{i-1},\delta))$ is below the stopping threshold, i.e. $w_{o_i}(\mathbf{y}(B_{i-1},\delta))\leq \epsilon\frac{\bar{d}}{r}(1-\bar{y}') $. Consequently, we can relate the marginal value estimate of $b_i$ to that of $o_i$:
\begin{equation*}
w_{b_i}(\mathbf{y}(B_{i-1},\delta)) \geq (1-\epsilon) w_{o_i}(\mathbf{y}(B_{i-1},\delta)) - \epsilon\frac{\bar{d}}{r}(1-\bar{y}').
\end{equation*}

Note that when $b_i$ is selected we have that $B_{i-1}+b_i\in \mathcal{I}$, and by the definition of $o_i$ (i.e. the matroid exchange property) $B_{i-1}+o_i\in \mathcal{I}$.

Now we need to bound the error incurred by sampling. From Lemma \ref{hoeffdingNonMon} we can sample the marginal values up to an additive error of $\beta = \frac{\epsilon}{r} f(OPT)(1-\bar{y}')$ by taking the average of $O\bigg(\frac{r^2}{\epsilon^2}\big(\frac{\bar{d}+\ubar{d}}{\bar{d}}\big)^2\log(|E|)\bigg)$ samples with high probability (i.e. with a bad estimate probability decreasing with $\frac{1}{|E|}$). (We do not have a term $O(\frac{1}{1-\bar{y}'})$ in the number of samples because our stopping time, $t=1$, and the update rule enforce (see Lemma \ref{boundY}) that $(1-\bar{y}')\geq \frac{1}{e}$). Thus, we can write:
\begin{equation*}
\begin{aligned}
w_{b_i}(\mathbf{y}(B_{i-1},\delta)) \leq \Delta F_{b_i}(\mathbf{y}(B_{i-1},\delta)) + \frac{\epsilon}{r} f(OPT)(1-\bar{y}') \\
w_{o_i}(\mathbf{y}(B_{i-1},\delta)) \geq \Delta F_{o_i}(\mathbf{y}(B_{i-1},\delta)) - \frac{\epsilon}{r} f(OPT)(1-\bar{y}')
\end{aligned}
\end{equation*}
which, with $\bar{d}\leq f(OPT)$, can be combined with the prior bound to yield:
\begin{equation}
\Delta F_{b_i}(\mathbf{y}(B_{i-1},\delta)) \geq (1-\epsilon) \Delta F_{o_i}(\mathbf{y}(B_{i-1},\delta)) - 3 \frac{\epsilon}{r} f(OPT)(1-\bar{y}').
\label{decThresBound}
\end{equation}

Then, we can bound the improvement that the Decreasing-Threshold procedure obtains:

\begin{align*}
&F(\mathbf{y}') - F(\mathbf{y}) \\
& =\sum_{i=1}^r ( F(\mathbf{y}(B_i,\delta)) -  F(\mathbf{y}(B_{i-1},\delta)) \\
& =\sum_{i=1}^r ( y_{b_i}' -  y_{b_i}) \frac{\partial F }{\partial y_{b_i}}\bigg|_{\mathbf{y}=\mathbf{y}(B_{i-1},\delta)}\\
& =\sum_{i=1}^r (1-e^{-\delta})(1-y_{b_i}) \frac{\partial F }{\partial y_{b_i}}\bigg|_{\mathbf{y}=\mathbf{y}(B_{i-1},\delta)}\\
& = (1-e^{-\delta})\sum_{i=1}^r \Delta F_{b_i}(\mathbf{y}(B_{i-1},\delta))\\
& \geq (1-e^{-\delta}) \sum_{i=1}^r \bigg( (1-\epsilon)\Delta F_{o_i}(\mathbf{y}(B_{i-1},\delta)) -3\frac{\epsilon}{r} f(OPT)(1-\bar{y}') \bigg) \\
& = (1-e^{-\delta})\bigg( (1-\epsilon)\sum_{i=1}^r \bigg( \Delta F_{o_i}(\mathbf{y}(B_{i-1},\delta)) \bigg)-3\epsilon f(OPT)(1-\bar{y}') \bigg)\\
& \geq (1-e^{-\delta}) \bigg( (1-\epsilon) \big( F(\mathbf{y}' \vee \mathbf{1}_{\text{OPT}}) - F(\mathbf{y}')\big) -3\epsilon f(OPT) (1-\bar{y}')\bigg)\\
& \geq (1-e^{-\delta}) \bigg( (1-4\epsilon)  (1-\bar{y}')f(OPT) - F(\mathbf{y}') \bigg).\\
\end{align*}

The second equality comes from the the multilinearity of $F$. With the update step $y' = 1+e^{-\delta}(y-1)$, we have that the increment is $y' - y = (1-e^{-\delta}) (1-y) $, which gives the third inequality. The fourth equality is by definition of $\Delta F_e$. The fifth inequality is by the bound in equation \ref{decThresBound}, the sixth by submodularity, and the last one by Lemma \ref{boundF}.
\end{proof}

Finally, we can use the above result to build a recurrence relation that yields the $\frac{1}{e}-\epsilon$ approximation ratio.

\begin{theorem}
The accelerated measured continuous greedy algorithm returns a point $\mathbf{y}^*\in P(\mathcal{M})$, such that:
\begin{equation}
F(\mathbf{y}^*)\geq \bigg(\frac{1}{e}-2\epsilon\bigg)f(OPT).
\end{equation}
\label{ratioAccMeas}
\end{theorem}

\begin{proof}
From Lemma \ref{deltaValDT} we have that:
\begin{equation*}
F(\mathbf{y}(t+\delta)) - F(\mathbf{y}(t))\geq (1-e^{-\delta}) \bigg( (1-4\epsilon)  (1-\bar{y}(t+\delta))f(OPT) - F(\mathbf{y}(t+\delta)) \bigg)
\label{dtimprovement}
\end{equation*}

We can now use the bound on the value of the coordinates of $\mathbf{y}$ in Lemma \ref{boundY},  we have that $y_{i}(t)\leq 1-e^{-t} \quad \forall  i \in E$, hence we can write:
\begin{equation}
F(\mathbf{y}(t+\delta)) - F(\mathbf{y}(t))\geq (1-e^{-\delta}) \bigg( (1-4\epsilon)  e^{-(\delta+t)}f(OPT) - F(\mathbf{y}(t+\delta)) \bigg).
\end{equation}

Consider the recurrence relation $a(n+1)-a(n)=k_1 (k_2 \exp (-(n+1) \epsilon )-a(n+1))$, which, with $a(0)=a_0$, has the following solution 
\begin{equation}
a(n)=\frac{(\frac{1}{k_1+1})^n (a_0 (-k_1+e^{\epsilon }-1)+k_1 k_2)-k_1 k_2 e^{-n \epsilon }}{-k_1+e^{\epsilon }-1}.
\end{equation}
This recurrence is equivalent to equation \ref{dtimprovement} if we set $k_1=(1-e^{-\delta})$, $k_2 = (1-4\epsilon)$, $n = \frac{t}{\delta}$, and $\delta = \epsilon$. Thus, substituting and simplifying assuming that $F(\mathbf{0})\geq 0$ (due to the non-negativity of $f$), we have the following lower bound on the value of the solution $\mathbf{y}(t)$ for any time $t\in[0,1]$:
\begin{equation*}
F(\mathbf{y}(t))\geq \frac{ (1-e^{-\epsilon }) (1-4 \epsilon ) e^{\epsilon  (-(\frac{t}{\epsilon }+1))} (e^{\epsilon  (\frac{t}{\epsilon }+1)} (\frac{1}{2-e^{-\epsilon }})^{t/\epsilon }-e^{\epsilon })}{e^{-\epsilon }+e^{\epsilon }-2} f(OPT).
\end{equation*}
So, when the algorithm ends at $t=1$, we have:
\begin{equation*}
F(\mathbf{y}(1))\geq \frac{ (1-e^{-\epsilon }) (1-4 \epsilon ) e^{\epsilon  (-(\frac{1}{\epsilon }+1))} (e^{\epsilon  (\frac{1}{\epsilon }+1)} (\frac{1}{2-e^{-\epsilon }})^{1/\epsilon }-e^{\epsilon })}{e^{-\epsilon }+e^{\epsilon }-2} f(OPT).
\end{equation*}
We can find a more intuitive version of this bound by observing that, clearly, for $0\leq \epsilon \leq 1$:
\begin{equation}
\frac{1}{e}-\frac{ (1-e^{-\epsilon }) (1-4 \epsilon ) e^{\epsilon  (-(\frac{1}{\epsilon }+1))} (e^{\epsilon  (\frac{1}{\epsilon }+1)} (\frac{1}{2-e^{-\epsilon }})^{1/\epsilon }-e^{\epsilon })}{e^{-\epsilon }+e^{\epsilon }-2} \leq \frac{5}{e}\epsilon \leq 2\epsilon
\end{equation}
Hence:
\begin{equation}
F(\mathbf{y}(1))\geq \bigg(\frac{1}{e}-2\epsilon\bigg)f(OPT).
\end{equation}
Finally, from Lemma \ref{boundY} we have that $\mathbf{y}(1) \in P(\mathcal{M})$.

\end{proof}

\subSectionInShort{Running Time}
\inArticle{As it is common in the submodular maximisation literature, we} quantify the running time in terms of value oracle calls to the submodular function and matroid independence oracle calls. We analyse the running time of the algorithm in two steps: first we study the running time of the Decreasing-Threshold procedure, and then that of the the continuous greedy.

\begin{lemma}
The Decreasing-Threshold procedure makes \\ $O\left(\frac{|E|r^2}{\epsilon^3}\log(|E|)\log(\frac{r}{\epsilon}) \left(\frac{\bar{d}+\ubar{d}}{\bar{d}}\right)^2\right)$
 value oracle calls, and $O\big(\frac{|E|}{\epsilon}\log{\frac{r}{\epsilon}}\big)
$ independence oracle calls.
\label{complexityOfDT}
\end{lemma}


\begin{proof}
First, the number of values that the threshold takes in the Decreasing-Threshold procedure to reach the stopping threshold is, considering that the term $(1-\bar{y})\geq \frac{1}{e}$ due to Lemma \ref{boundY}, $O(\frac{\log{\frac{\epsilon}{r}}}{\log{(1-\epsilon)}})$. Second, for each threshold value, the algorithm performs $O(|E|)$ estimates of $\Delta F_e$, and $O(|E|)$ calls to the independence oracle. Therefore, the number of independence oracle calls is the number of calls per threshold step multiplied by the number of threshold steps, i.e.:
\begin{equation}
O\bigg(\frac{|E|}{\epsilon}\log{\frac{r}{\epsilon}}\bigg),
\end{equation}
which has been simplified noting that $\frac{\log{\frac{\epsilon}{r}}}{\log{1-\epsilon}}\leq \frac{1}{\epsilon }\log (\frac{r}{\epsilon })$.\\
Now, each estimate of $\Delta F_e$ requires $O\bigg(\frac{r^2}{\epsilon^2}\big(\frac{\bar{d}+\ubar{d}}{\bar{d}}\big)^2\log(|E|)\bigg)$ samples. Hence, we can conclude that the number of value oracle calls is 
\begin{equation*}
O\left(\frac{|E|r^2}{\epsilon^3}\log(|E|)\log \left(\frac{r}{\epsilon} \right) \left(\frac{\bar{d}+\ubar{d}}{\bar{d}}\right)^2\right).
\end{equation*}
\end{proof}

We can now quantify the running time of the whole algorithm.

\begin{theorem}
The accelerated measured continuous greedy algorithm makes \\ $O\left(\frac{|E|r^2}{\epsilon^4}\log(|E|)\log(\frac{r}{\epsilon}) \left(\frac{\bar{d}+\ubar{d}}{\bar{d}}\right)^2\right)$
 value oracle calls, and $O\big(\frac{|E|}{\epsilon^2}\log{\frac{r}{\epsilon}}\big)
$ independence oracle calls.
\end{theorem}

\begin{proof}
Considering that the number of steps of the procedure, with $\delta = \epsilon$, is $\frac{1}{\epsilon}$. The number of calls to boths oracles is simply the number of calls by the Decreasing-Threshold procedure (Algorithm \ref{algoThreshold}) in Lemma \ref{complexityOfDT} multiplied by the number of steps in the continuous greedy algorithm (Algorithm \ref{algoThresholdMeasured}).
\end{proof}

\sectionInShort{Discussion and Future Work}
\inArticle{In this paper we} have presented a $\frac{1}{e}-\epsilon$-approximation algorithm for general non-negative submodular function maximisation that requires $O(\frac{nr^2}{\epsilon^4} \big(\frac{\bar{d}+\ubar{d}}{\bar{d}}\big)^2 \log^2({\frac{n}{\epsilon}}))$ value oracle calls. This is the fastest $\frac{1}{e}$-approximation algorithm currently available, which enables the use of general (non-monotone) matroid-constrained submodular maximisation for many applications for which existing algorithms were implausibly slow. We think this is of great significance because there has been a recent surge of interest for applying submodular maximisation in fields where large problem instances are paramount, such as Machine Learning \cite{mirzasoleiman2016fast, Bilmes2017, lin2010multi}, particularly in the field of summarisation where non-monotone submodular functions are natural \cite{tschiatschek2014learning, dasgupta2013summarization}. Our algorithm is slower than the one presented for the monotone case in \cite{FastVondrak} by $O\big(r\big(\frac{\bar{d}+\ubar{d}}{\bar{d}}\big)^2\big)$  due to the inability to sample the marginal values of the multilinear extension up to an additive and multiplicative bound. If we could, then we would reduce the additional value oracle calls required to achieve an algorithm with the same running time, we believe this might be possible. 

A future avenue of research would be to combine our work with the very interesting results in \cite{Buchbinder2014}, where an efficient algorithm is proposed to allow the trade-off of value oracle calls and matroid independence calls for non-negative monotone submodular functions, to enable query trade-off for general non-negative submodular functions. Another interesting path is to combine the more continuous-like measured continuous greedy update step that we present here with the acceleration techniques for strong submodular functions presented in \cite{Wang2015} to produce an adaptive step algorithm. This way, in each step we could use a large $\delta$ that extended to the boundary of the region of validity of the set $B$, instead of taking a $\delta$ that is small enough to satisfy the worst case. Another obvious improvement on the algorithms presented here would be to combine the ideas from the Lazy Greedy Algorithm \cite{Minoux1978} to adaptively change the decrement of the threshold in the Decreasing-Threshold procedure. 

%% file: AcceleratedMeasuredArticle.bbl
\providecommand{\href}[2]{#2}\begingroup\raggedright\endgroup

%% file: AcceleratedMeasuredArticle.bbl
\begin{thebibliography}{10}

\bibitem{krause2012submax}
A.~Krause and D.~Golovin, ``{Submodular function maximization},'' {\em
  Tractability: Practical Approaches to Hard Problems} {\bfseries 3} no.~19,
  (2012) 8.

\bibitem{lehmann2001comb}
B.~Lehmann, D.~Lehmann, and N.~Nisan, ``{Combinatorial auctions with decreasing
  marginal utilities},'' in {\em Proceedings of the 3rd ACM conference on
  Electronic Commerce}, pp.~18--28, ACM.
\newblock 2001.

\bibitem{dughmi2009revenue}
S.~Dughmi, T.~Roughgarden, and M.~Sundararajan, ``{Revenue submodularity},'' in
  {\em Proceedings of the 10th ACM conference on Electronic commerce},
  pp.~243--252, ACM.
\newblock 2009.

\bibitem{kempe2003social}
D.~Kempe, J.~Kleinberg, and {\'{E}}.~Tardos, ``{Maximizing the spread of
  influence through a social network},'' in {\em Proceedings of the ninth ACM
  SIGKDD international conference on Knowledge discovery and data mining},
  pp.~137--146, ACM.
\newblock 2003.

\bibitem{lin2010multi}
H.~Lin and J.~Bilmes, ``{Multi-document summarization via budgeted maximization
  of submodular functions},'' in {\em Human Language Technologies: The 2010
  Annual Conference of the North American Chapter of the Association for
  Computational Linguistics}, pp.~912--920, Association for Computational
  Linguistics.
\newblock 2010.

\bibitem{krause2008efficient}
A.~Krause, J.~Leskovec, C.~Guestrin, J.~VanBriesen, and C.~Faloutsos,
  ``{Efficient sensor placement optimization for securing large water
  distribution networks},'' {\em Journal of Water Resources Planning and
  Management} {\bfseries 134} no.~6, (2008) 516--526.

\bibitem{leskovec2007cost}
J.~Leskovec, A.~Krause, C.~Guestrin, C.~Faloutsos, J.~VanBriesen, and
  N.~Glance, ``{Cost-effective outbreak detection in networks},'' in {\em
  Proceedings of the 13th ACM SIGKDD international conference on Knowledge
  discovery and data mining}, pp.~420--429, ACM.
\newblock 2007.

\bibitem{Schrijver2003}
A.~Schrijver, {\em {Combinatorial Optimization: Polyhedra and Efficiency}}.
\newblock Algorithms and Combinatorics. Springer Berlin Heidelberg, 2003.

\bibitem{Nemhauser1978}
G.~L. Nemhauser, L.~A. Wolsey, and M.~L. Fisher, ``{An analysis of
  approximations for maximizing submodular set functions—I},''
  \href{http://dx.doi.org/10.1007/BF01588971}{{\em Mathematical Programming}
  {\bfseries 14} no.~1, (Dec, 1978) 265--294}.

\bibitem{Vondrak2008}
J.~Vondrak, \href{http://dx.doi.org/10.1145/1374376.1374389}{``{Optimal
  approximation for the submodular welfare problem in the value oracle
  model},''} in {\em Proceedings of the fourtieth annual ACM symposium on
  Theory of computing - STOC 08}, p.~67.
\newblock ACM Press, New York, New York, USA, May, 2008.

\bibitem{Calinescu2011}
G.~Calinescu, C.~Chekuri, M.~P{\'{a}}l, and J.~Vondr{\'{a}}k, ``{Maximizing a
  Monotone Submodular Function Subject to a Matroid Constraint},''
  \href{http://dx.doi.org/10.1137/080733991}{{\em SIAM Journal on Computing}
  {\bfseries 40} no.~6, (Jan, 2011) 1740--1766}.

\bibitem{Feldman2011}
M.~Feldman, J.~Naor, and R.~Schwartz,
  \href{http://dx.doi.org/10.1109/FOCS.2011.46}{``{A Unified Continuous Greedy
  Algorithm for Submodular Maximization},''} in {\em 2011 IEEE 52nd Annual
  Symposium on Foundations of Computer Science}, pp.~570--579.
\newblock IEEE, Oct, 2011.

\bibitem{vondrak2011submodular}
J.~Vondr{\'{a}}k, C.~Chekuri, and R.~Zenklusen, ``{Submodular function
  maximization via the multilinear relaxation and contention resolution
  schemes},'' in {\em Proceedings of the forty-third annual ACM symposium on
  Theory of computing}, pp.~783--792, ACM.
\newblock 2011.

\bibitem{FastVondrak}
A.~Badanidiyuru and J.~Vondr{\'{a}}k, ``{Fast algorithms for maximizing
  submodular functions},'' in {\em Proceedings of the Twenty-Fifth Annual
  ACM-SIAM Symposium on Discrete Algorithms}, pp.~1497--1514, Society for
  Industrial and Applied Mathematics.
\newblock 2014.

\bibitem{feige1998threshold}
U.~Feige, ``{A threshold of ln n for approximating set cover},'' {\em Journal
  of the ACM (JACM)} {\bfseries 45} no.~4, (1998) 634--652.

\bibitem{gharan2011submodular}
S.~O. Gharan and J.~Vondr{\'{a}}k, ``{Submodular maximization by simulated
  annealing},'' in {\em Proceedings of the twenty-second annual ACM-SIAM
  symposium on Discrete Algorithms}, pp.~1098--1116, Society for Industrial and
  Applied Mathematics.
\newblock 2011.

\bibitem{ene2016constrained}
A.~Ene and H.~L. Nguyen, ``{Constrained submodular maximization: Beyond 1/e},''
  in {\em Foundations of Computer Science (FOCS), 2016 IEEE 57th Annual
  Symposium on}, pp.~248--257, IEEE.
\newblock 2016.

\bibitem{buchbinder2016constrained}
N.~Buchbinder and M.~Feldman, ``{Constrained Submodular Maximization via a
  Non-symmetric Technique},'' \href{http://arxiv.org/abs/1611.03253}{{\ttfamily
  arXiv:1611.03253}}.

\bibitem{Hoeffding}
W.~Hoeffding, ``{Probability inequalities for sums of bounded random
  variables},'' {\em Journal of the American statistical association}
  {\bfseries 58} no.~301, (1963) 13--30.

\bibitem{Chekuri2014}
C.~Chekuri, J.~Vondr{\'{a}}k, and R.~Zenklusen, ``{Submodular Function
  Maximization via the Multilinear Relaxation and Contention Resolution
  Schemes},'' \href{http://dx.doi.org/10.1137/110839655}{{\em SIAM Journal on
  Computing} {\bfseries 43} no.~6, (Nov, 2014) 1831--1879}.

\bibitem{Ageev2004}
A.~Ageev and M.~Sviridenko, ``{Pipage Rounding: A New Method of Constructing
  Algorithms with Proven Performance Guarantee},''
  \href{http://dx.doi.org/10.1023/B:JOCO.0000038913.96607.c2}{{\em Journal of
  Combinatorial Optimization} {\bfseries 8} no.~3, (Sep, 2004) 307--328}.

\bibitem{vondrak2013symmetry}
J.~Vondr{\'{a}}k, ``{Symmetry and approximability of submodular maximization
  problems},'' {\em SIAM Journal on Computing} {\bfseries 42} no.~1, (2013)
  265--304.

\bibitem{Vondrak2010}
J.~Vondrak, ``{A note on concentration of submodular functions},''
  \href{http://arxiv.org/abs/1005.2791}{{\ttfamily arXiv:1005.2791}}.

\bibitem{mirzasoleiman2016fast}
B.~Mirzasoleiman, A.~Badanidiyuru, and A.~Karbasi, ``{Fast constrained
  submodular maximization: Personalized data summarization},'' in {\em ICLM'16:
  Proceedings of the 33rd International Conference on Machine Learning (ICML)}.
\newblock 2016.

\bibitem{Bilmes2017}
J.~Bilmes and W.~Bai, ``{Deep Submodular Functions},''
  \href{http://arxiv.org/abs/1701.08939}{{\ttfamily arXiv:1701.08939}}.

\bibitem{tschiatschek2014learning}
S.~Tschiatschek, R.~K. Iyer, H.~Wei, and J.~A. Bilmes, ``{Learning mixtures of
  submodular functions for image collection summarization},'' in {\em Advances
  in neural information processing systems}, pp.~1413--1421.
\newblock 2014.

\bibitem{dasgupta2013summarization}
A.~Dasgupta, R.~Kumar, and S.~Ravi, ``{Summarization Through Submodularity and
  Dispersion.},'' in {\em ACL (1)}, pp.~1014--1022.
\newblock 2013.

\bibitem{Buchbinder2014}
N.~Buchbinder, M.~Feldman, and R.~Schwartz, ``{Comparing Apples and Oranges:
  Query Tradeoff in Submodular Maximization},''
  \href{http://arxiv.org/abs/1410.0773}{{\ttfamily arXiv:1410.0773}}.

\bibitem{Wang2015}
Z.~Wang, B.~Moran, X.~Wang, and Q.~Pan, ``{An accelerated continuous greedy
  algorithm for maximizing strong submodular functions},''
  \href{http://dx.doi.org/10.1007/s10878-013-9685-x}{{\em Journal of
  Combinatorial Optimization} {\bfseries 30} no.~4, (2015) 1107--1124}.

\bibitem{Minoux1978}
M.~Minoux, \href{http://dx.doi.org/10.1007/BFb0006528}{``{Accelerated greedy
  algorithms for maximizing submodular set functions},''} in {\em Optimization
  Techniques}, pp.~234--243.
\newblock Springer-Verlag, Berlin/Heidelberg, 1978.

\end{thebibliography}
